\documentclass[12pt,reqno,fleqn]{amsart}
\usepackage{latexsym}
\usepackage{amssymb}
\usepackage{amsxtra}
\usepackage{amsmath}
\usepackage{amsfonts}
\usepackage[dvips]{graphicx}
\usepackage{amsmath}
\usepackage{amssymb,amsmath,amsthm,amscd,mathrsfs,amsbsy,amsfonts}
\usepackage{pstricks,pst-node}
\usepackage{amsbsy,young,colortbl,cite}

\usepackage{graphicx}
\usepackage[metapost]{mfpic}
\usepackage{mflogo}

\usepackage{epstopdf}
\usepackage{graphics}
\makeatletter      
\@addtoreset{equation}{section}
\makeatother       

\def\@makecaption#1#2{\vskip\abovecaptionskip
  \sbox\@tempboxa{\small #1: #2}%
  \ifdim \wd\@tempboxa >\hsize \small #1: #2\par
  \else \global \@minipagefalse \hb@xt@\hsize{\hfil\box\@tempboxa\hfil}\fi
  \vskip\belowcaptionskip}
\newtheorem{theorem}{Theorem}[section]
\newtheorem{lemma}[theorem]{Lemma}
\newtheorem{proposition}[theorem]{Proposition}

\newtheorem{definition}[theorem]{Definition}

\theoremstyle{remark}

\newcommand{\Exp}[1]{\operatorname{e}^{#1}}
\newcommand{\C}{\mathbb{C}}

\newcommand{\W}{\mathcal{W}}
\newcommand{\M}{\mathcal{M}}
\renewcommand{\L}{\mathcal{L}}

\newcommand{\N}{\mathbb{N}}

\newcommand{\g}{\mathfrak{g}}

\newcommand{\I}{\mathbb{I}}
\newcommand{\bS}{\mathbb{S}}
\newcommand{\bP}{\mathbb{P}}

\def\({\left(}
\def\){\right)}
\def\[{\begin{eqnarray}}
\def\]{\end{eqnarray}}
\def\d{\partial}

\def\La{\Lambda}

\begin{document}

\title[Strongly coupled extended  Toda]{A strongly coupled extended  Toda hierarchy and its Virasoro symmetry}
\author{Chuanzhong Li}
\address{School of Mathematics and Statistics, Ningbo University,
Ningbo, 315211 Zhejiang, P.\ R.\ China}
\email{lichuanzhong@nbu.edu.cn}


\date{}

\begin{abstract}
As a generalization of the integrable  extended Toda hierarchy and a reduction of the extended multicomponent Toda hierarchy,  from the point of a commutative subalgebra of $gl(2,\C)$, we construct a strongly coupled extended Toda hierarchy(SCETH) which will be proved to possess a Virasoro type additional symmetry by acting on its tau-function. Further we give the multi-fold Darboux transformations of the strongly coupled extended  Toda hierarchy.
\end{abstract}

\maketitle

\allowdisplaybreaks
\maketitle{\allowdisplaybreaks}
{\small \noindent {\bf Mathematics Subject Classifications (2010)}:  37K05, 37K10, 37K20.

}
 \setcounter{section}{0}
\section{Introduction}
In the theory of integrable systems, two important fundamental models are the KP and Toda systems from which one can derive a lot of local and nonlocal integrable equations \cite{NLKP,NLKP2,NLKP3,NLTODA,NLTODA2}.
The  Toda lattice hierarchy
  as a completely integrable system  has many important applications in mathematics and physics including the representation theory of Lie algebras, orthogonal polynomials and  random
matrix models  \cite{Toda,Todabook,UT,witten,dubrovin}. The  Toda system has many kinds of reductions or extensions such as the extended Toda hierarchy (ETH)\cite{CDZ,M}, bigraded Toda hierarchy (BTH)\cite{C}-\cite{dispBTH} and so on.
Considering the application in the Gromov-Witten theory, the Toda hierarchy was extended to the extended Toda hierarchy\cite{CDZ} which governs the Gromov-Witten invariants of $\C P^1$.
 The extended bigraded Toda
hierarchy(EBTH) is the extension of the bigraded Toda
hierarchy (BTH) which includes  additional logarithmic flows\cite{C} with considering its application in the Gromov-Witten theory of orbifolds $C_{N,M}$.
  In \cite{EMTH}, the extended flow equations of the multi-component  Toda hierarchy were constructed.  Meanwhile the Darboux transformation and bi-Hamiltonian structure of this new extended multi-component  Toda hierarchy(EMTH) were given.
We considered the Hirota quadratic equation of the commutative subhierarchy of the extended multi-component Toda hierarchy which might be useful in the theory of Frobenius manifolds  in \cite{EZTH}.
In \cite{EZTH}, we constructed the extended flow equations of a new $Z_N$-Toda hierarchy which took values in a commutative subalgebra $Z_N$ of $gl(N,\mathbb C)$. Meanwhile we gave the Hirota bilinear equations and tau functions of the hierarchy which might be useful in the topological field theory and Gromov-Witten theory.

   The additional symmetry is a universal property for integrable systems \cite{os1,dickey1,dickey2}.  Among the algebraic structures of additional symmetries, the Virasoro symmetry is one kind of important symmetries such as \cite{cmp}.
  In \cite{ourBlock}, we provided  a kind of Block type algebraic structures for the bigraded Toda hierarchy (BTH) \cite{ourJMP,solutionBTH}.
Later on, this Block type Lie algebra was found again in the dispersionless bigraded Toda
 hierarchy \cite{dispBTH}.

It was pointed out that the Darboux transformation was an efficient method to generate soliton solutions of integrable equations.  The multi-solitons can be obtained by this Darboux transformation from a trivial seed solution.

This paper will be arranged as follows. In the next section we recall the extended multicomponent Toda hierarchy. In Section 3,
we will give the additional symmetry of the extended multicomponent Toda hierarchy which constitutes a Virasoro type Lie algebra. From the point of a commutative reduction from Lie algebras, the strongly coupled extended Toda hierarchy is recalled in
Section 4. The additional symmetry of the SCETH will be constructed in Section 5 and this symmetry has a Virasoro type structure which includes the Virasoro algebra as a subalgebra.
 The Virasoro action on the tau function of the SCETH will be given in Section 6. Further after that, we give its multi-fold Darboux transformations of the strongly coupled extended  Toda hierarchy.

\section{Extended multicomponent Toda hierarchy}
In this section by following \cite{EMTH}, we will denote $G_m$ as a group which contains invertible elements of $N\times N$
complex matrices and denote its Lie algebra  $\g_m$ as the associative algebra  of  $N\times N$
complex matrices $M_N(\C)(\Lambda)$ where  the shift operator $\Lambda$ acting on any functions $g(x)$ as
$(\Lambda g)(x):=g(x+\epsilon)$.

In this section, firstly let us recall the basic notation of the extended multicomponent Toda hierarchy defined in \cite{EMTH}.

 In \cite{EMTH}, we   define the dressing operators $W,\bar W$ as follows
\begin{align}
\label{def:baker}W&:=S\cdot W_0,& \bar W&:=\bar S\cdot \bar  W_0,
\end{align}
where $S,\bar S$ have expansions of the form
\begin{gather}
\label{expansion-S}
\begin{aligned}
S&=\I_N+\omega_1(x)\Lambda^{-1}+\omega_2(x)\Lambda^{-2}+\cdots\in G_{m-},\\
\bar S&=\bar\omega_0(x)+\bar\omega_1(x)\Lambda+\bar\omega_2(x)\Lambda^{2}+\cdots\in
G_{m+},
\end{aligned}
\end{gather}
and $ G_{m-}$ and  $ G_{m+}$ are two subgroups of $G_{m}$.
The  free operators $ W_0,\bar  W_0\in G_m$ have forms as
\begin{align}
 \label{def:E}  W_0&:=\sum_{k=1}^NE_{kk}\exp\left(\sum_{j=0}^\infty
 t_{jk}\Lambda^{j}+ s_{j}\frac{\Lambda^{j}}{j!}(\epsilon\partial-c_j)\right), \\
\label{def:barE}   \bar W_0&:=\sum_{k=1}^NE_{kk}\exp\left(\sum_{j=0}^\infty\bar
   t_{j k}\Lambda^{-j}+ s_{j}\frac{\Lambda^{-j}}{j!}(\epsilon\partial-c_j)\right),
\end{align}
with $t_{jk}, \bar t_{jk},s_{j} $
 as continuous times.
Also we define the symbols of $S,\bar S$ as  $\bS,\bar \bS$
\begin{gather}
\begin{aligned}
\bS&=\I_N+\omega_1(x)\lambda^{-1}+\omega_2(x)\lambda^{-2}+\cdots,\\
\bar \bS&=\bar\omega_0(x)+\bar\omega_1(x)\lambda+\bar\omega_2(x)\lambda^{2}+\cdots.
\end{aligned}
\end{gather}

Also the inverse operators $S^{-1},\bar S^{-1}$ of operators $S,\bar S$ have expansions of the form
\begin{gather}
\begin{aligned}
S^{-1}&=\I_N+\omega'_1(x)\Lambda^{-1}+\omega'_2(x)\Lambda^{-2}+\cdots\in G_{m-},\\
\bar S^{-1}&=\bar\omega'_0(x)+\bar\omega'_1(x)\Lambda+\bar\omega'_2(x)\Lambda^{2}+\cdots\in
G_{m+}.
\end{aligned}
\end{gather}
Also we define the symbols of $S^{-1},\bar S^{-1}$  as  $\bS^{-1},\bar \bS^{-1}$
\begin{gather}
\begin{aligned}
\bS^{-1}&=\I_N+\omega'_1(x)\lambda^{-1}+\omega'_2(x)\lambda^{-2}+\cdots,\\
\bar \bS^{-1}&=\bar\omega'_0(x)+\bar\omega'_1(x)\lambda+\bar\omega'_2(x)\lambda^{2}+\cdots.
\end{aligned}
\end{gather}

 The Lax  operators $L,C_{kk},\bar C_{kk}\in\g_m$
 are defined by
\begin{align}
\label{Lax}  L&:=W\cdot\Lambda\cdot W^{-1}=\bar W\cdot\Lambda^{-1}\cdot \bar W^{-1}, \\
\label{C} C_{kk}&:=W\cdot E_{kk}\cdot W^{-1},& \bar C_{kk}&:=\bar
W\cdot E_{kk}\cdot \bar W^{-1},
\end{align}
and
have the following expansions
\begin{gather}\label{lax expansion}
\begin{aligned}
 L&=\Lambda+u_1(x)+u_2(x)\Lambda^{-1}, \\
C_{kk}&=E_{kk}+C_{kk,1}(x)\Lambda^{-1}+C_{kk,2}(x)\Lambda^{-2}+\cdots,\\
\bar C_{kk}&=\bar C_{kk,0}(x)+\bar C_{kk,1}(x)\Lambda+\bar
C_{kk,2}(x)\Lambda^{2}+\cdots.
\end{aligned}
\end{gather}
 In fact the Lax  operators $L,C_{kk},\bar C_{kk}\in\g_m$
 can also be equivalently defined by
\begin{align}
\label{two dressing}  L&:=S\cdot\Lambda\cdot S^{-1}=\bar S\cdot\Lambda^{-1}\cdot \bar S^{-1}, \\
\label{C} C_{kk}&:=S\cdot E_{kk}\cdot S^{-1},& \bar C_{kk}&:=\bar
S\cdot E_{kk}\cdot \bar S^{-1}.
\end{align}

The matrix operators $B_{jk},\bar B_{jk},D_{j}$ are defined as follows
\begin{align}\label{satoS}
\begin{aligned}
B_{jk}&:=WE_{kk}\Lambda^jW^{-1},\ \ \bar B_{jk}:=\bar W E_{kk}\Lambda^{-j}\bar W^{-1},\\
D_{j}&:=\frac{2\L^j}{j!}(\log \L-c_j),\ \ c_0=0;\ c_j=\sum_{i=1}^{j}\frac1i,j\geq 1.
\end{aligned}
\end{align}

To define extended flows of the extended multi-component  Toda hierarchy(EMTH), we define the following logarithmic matrices \cite{EMTH}

\begin{align}
\log_+L&=(S\cdot\epsilon \partial\cdot S^{-1})= \epsilon \partial+\sum_{k < 0} W_k(x) \Lambda^k,\\
\log_-L&=-(\bar S\cdot\epsilon \partial\cdot \bar S^{-1})=-\epsilon \partial+\sum_{k \geq 0} W_k(x) \Lambda^k,
\end{align}
where $\d$ is the derivative with respect to the spatial variable $x$.
Combining these above logarithmic operators together can help us in deriving the following important logarithmic matrix
\begin{align}
\label{Log} \log L:&=\frac12\log_+L+\frac12\log_-L=\frac12(S\cdot\epsilon \partial\cdot S^{-1}-\bar S\cdot\epsilon \partial\cdot \bar S^{-1}),
\end{align}
which will generate a series of extended flow equations  contained in the following Lax equations.
\begin{proposition}\label{Lax}
 The  Lax equations of the EMTH are as follows
   \begin{align}
\label{laxtjk}
  \epsilon\partial_{t_{jk}} L&= [(B_{jk})_+,L],&
 \epsilon\partial_{t_{jk}} C_{ss}&= [-(B_{jk})_-,C_{ss}],&\epsilon\partial_{t_{jk}} \bar C_{ss}&= [(B_{jk})_+,\bar C_{ss}],
\\
  \epsilon\partial_{\bar t_{jk}} L&= [ (\bar B_{jk})_+,L],&
 \epsilon\partial_{\bar t_{jk}} C_{ss}&= [-(\bar B_{jk})_-,C_{ss}],&\epsilon\partial_{\bar t_{jk}} \bar C_{ss}&= [(\bar B_{jk})_+,\bar C_{ss}],
\\
 \epsilon\partial_{s_{j}} L&= [(D_{j})_+,L],&
  \epsilon\partial_{s_{j}} C_{ss}&= [-(D_{j})_-,C_{ss}],& \epsilon\partial_{s_{j}} \bar C_{ss}&= [(D_{j})_+,\bar C_{ss}].
  \end{align}
\end{proposition}

\section{Strongly coupled extended Toda hierarchy}
In this section, we firstly construct a strongly coupled extended Toda hierarchy.
The algebra has a maximal symmetric commutative subalgebra $S_2=\C[\Gamma]$ and $\Gamma=\begin{bmatrix}0&1\\ 1&0\end{bmatrix}\in gl(2,\C).$
Denote $\Lambda$ as a shift operator by acting on any function $f(x)$ as $\Lambda f(x)=f(x+\epsilon)$,  and an algebra $S_2(\Lambda):=\g$, then the algebra $\g$
 has the following splitting
\begin{gather}\label{splittingc}
\g=\g_{+}\oplus\g_{-},
\end{gather}
where
\begin{align*}
  \g_{+}&=\Big\{\sum_{j\geq 0}X_j(x)\Lambda^j,\quad X_j(x)\in S_2\Big\},&
  \g_{-}&=\Big\{\sum_{j< 0}X_j(x)\Lambda^j,\quad X_j(x)\in S_2\Big\}.
\end{align*}
The splitting \eqref{splittingc} leads us to consider the following factorization of
$g\in G$
\begin{gather}\label{fac1}
g=g_{-}^{-1}\circ g_{+}, \quad g_{\pm}\in G_{\pm},
\end{gather}
where $G_{\pm}$ have $\g_{\pm}$ as their Lie algebras. $G_{+}$
is the set of invertible linear operators  of the
form $\sum_{j\geq 0}g_j(x)\Lambda^j$; while $G_{-}$ is the set of
invertible linear operators of the form
$\I_2+\sum_{j<0}g_j(x)\Lambda^j$.

 Now we
introduce  the following free operators $ \W_0,\bar  \W_0\in G$
\begin{align}
 \label{def:E}  \W_{0}&:=\exp{\sum_{j=0}^\infty
 t_{j}\frac{\Lambda^j}{\epsilon j!}+ y_{j}\frac{\Lambda^j}{\epsilon j!}(\epsilon \partial-c_j)},\ \ \partial=\frac{\partial}{\partial  x}, \\
\label{def:barE}   \bar \W_{0}&:=\exp{\sum_{j=0}^\infty
   t_{j}\frac{\Lambda^{-j}}{\epsilon j!}+ y_{j}\frac{\Lambda^{-j}}{\epsilon j!}(\epsilon \partial-c_j)}, \ \ c_j=\sum_{i=1}^j\frac 1i,
\end{align}
where $t_{j}, y_{j} \in \C$
will play the role of continuous times.

 We   define the dressing operators $\W,\bar \W$ as follows
\begin{align}
\label{def:baker}\W&:=P\circ \W_0,\ \  \bar \W:=\bar P\circ \bar  \W_0,\quad P\in G_{-},\ \bar P\in G_{+}.
\end{align}
Given an element $g\in G$ and denote $t=(t_{j}), y=(y_{j}); j\mathbb\in \N$, one can consider the factorization problem  in $G$
\begin{gather}
  \label{facW}
  \W\circ g=\bar \W,
\end{gather}
i.e.
 the factorization problem
\begin{gather}
  \label{factorization}
  P(t,y)\circ \W_0\circ g=\bar P(t,y)\circ\bar \W_0.
\end{gather}
Observe that  $P,\bar P$ have expansions of the form
\begin{gather}
\label{expansion-S}
\begin{aligned}
P&=\I_2+\omega_1(x)\Lambda^{-1}+\omega_2(x)\Lambda^{-2}+\cdots\in G_{-},\\
\bar P&=\bar\omega_0(x)+\bar\omega_1(x)\Lambda+\bar\omega_2(x)\Lambda^{2}+\cdots\in
G_{+}.
\end{aligned}
\end{gather}
Also we define the symbols of $P,\bar P$ as  $\bP,\bar \bP$
\begin{gather}
\begin{aligned}
\bP&=\I_2+\omega_1(x)\lambda^{-1}+\omega_2(x)\lambda^{-2}+\cdots=\begin{bmatrix}\bP_0&\bP_1\\ \bP_1&\bP_0\end{bmatrix},\\
\bar \bP&=\bar\omega_0(x)+\bar\omega_1(x)\lambda+\bar\omega_2(x)\lambda^{2}+\cdots=\begin{bmatrix}\bar \bP_0&\bar \bP_1\\\bar \bP_1&\bar \bP_0\end{bmatrix}.
\end{aligned}
\end{gather}

The inverse operators $P^{-1},\bar P^{-1}$ of operators $P,\bar P$ have expansions of the form
\begin{gather}
\begin{aligned}
P^{-1}&=\I_2+\omega'_1(x)\Lambda^{-1}+\omega'_2(x)\Lambda^{-2}+\cdots\in G_{-},\\
\bar P^{-1}&=\bar\omega'_0(x)+\bar\omega'_1(x)\Lambda+\bar\omega'_2(x)\Lambda^{2}+\cdots\in
G_{+}.
\end{aligned}
\end{gather}
Also we define the symbols of $P^{-1},\bar P^{-1}$  as  $\bP^{-1},\bar \bP^{-1}$
\begin{gather}
\begin{aligned}
\bP^{-1}&=\I_2+\omega'_1(x)\lambda^{-1}+\omega'_2(x)\lambda^{-2}+\cdots=\begin{bmatrix}\bP_2&\bP_3\\ \bP_3&\bP_2\end{bmatrix},\\
\bar \bP^{-1}&=\bar\omega'_0(x)+\bar\omega'_1(x)\lambda+\bar\omega'_2(x)\lambda^{2}+\cdots=\begin{bmatrix}\bar\bP_2&\bar\bP_3\\ \bar\bP_3&\bar\bP_2\end{bmatrix}.
\end{aligned}
\end{gather}

 The Lax  operators $\L\in G$
 are defined by
\begin{align}
\label{LaxZ}  \L&:=\W\circ\Lambda\circ \W^{-1}=\bar \W\circ\Lambda^{-1}\circ \bar \W^{-1}=P\circ\Lambda\circ P^{-1}=\bar P\circ\Lambda^{-1}\circ \bar P^{-1},
\end{align}
and
have the following expansions
\begin{gather}\label{lax expansion}
\begin{aligned}
 \L&=\Lambda+u(x)+v(x)\Lambda^{-1}=\begin{bmatrix}\Lambda+u_0(x)+v_0(x)\Lambda^{-1}&u_1(x)+v_1(x)\Lambda^{-1}\\ u_1(x)+v_1(x)\Lambda^{-1}&\Lambda+u_0(x)+v_0(x)\Lambda^{-1}\end{bmatrix}.
\end{aligned}
\end{gather}

Now we define the following two logarithm matrices
\begin{align}
\log_+\L&=\W\circ\epsilon \partial\circ \W^{-1}=P\circ\epsilon \partial\circ P^{-1},\\
\log_-\L&=-\bar \W\circ\epsilon \partial\circ \bar \W^{-1}=-\bar P\circ\epsilon \partial\circ \bar P^{-1}.
\end{align}

Combining these above logarithm operators together can derive following important logarithm matrix
\begin{align}
\label{Log} \log \L:&=\frac12(\log_+\L+\log_-\L)=\frac12(P\circ\epsilon \partial\circ P^{-1}-\bar P\circ\epsilon \partial\circ \bar P^{-1}):=\sum_{i=-\infty}^{+\infty}\W_i\Lambda^i\in G,
\end{align}
which will generate a series of flow equations which contain the spatial flow in later defined Lax equations.

Let us first introduce some convenient notations of $S_2$-valued matrix operators $B_{j},F_{j}$ as follows
\begin{align}\label{satoS}
\begin{aligned}
B_{j}&:=\frac{\L^{j+1}}{(j+1)!},\ \
F_{j}:=\frac{2\L^j}{j!}(\log \L-c_j),\ \  c_j=\sum_{i=1}^j\frac 1i,\ j\geq 0.
\end{aligned}
\end{align}

Now we give the definition of the strongly coupled extended Toda hierarchy(SCETH).
\begin{definition}The strongly coupled extended Toda hierarchy is a hierarchy in which the dressing operators $P,\bar P$ satisfy the following Sato equations
\begin{align}
\label{satoStz} \epsilon\partial_{t_{j}}P&=-(B_{j})_-P,& \epsilon\partial_{t_{j}}\bar P&=(B_{j})_+\bar P,  \\
\label{satoSsz}\epsilon\partial_{ y_{j}}P&=-(F_{j})_- P,& \epsilon\partial_{y_{j}}\bar P&=(F_{j})_+\bar P,\end{align}
which is equivalent to that the dressing operators $W,\bar W$ are subject to the following Sato equations
\begin{align}
\label{Wjkz} \epsilon\partial_{t_{j}}W&=(B_{j})_+ W,& \epsilon\partial_{t_{j}}\bar W&=(B_{j})_+\bar W,  \\
\epsilon\partial_{y_{j}}W&=(\frac{\L^j}{j!}(\log_+ \L-c_j) -(F_{j})_-) W,& \epsilon\partial_{y_{j}}\bar W&=(-\frac{\L^j}{j!}(\log_- \L-c_j)+(F_{j})_+)\bar W.  \end{align}
\end{definition}

 From the previous proposition we derive the following  Lax equations for the Lax operators.
\begin{proposition}\label{Lax}
 The  Lax equations of the SCETH are as follows
   \begin{align}
\label{laxtjk}
  \epsilon\partial_{t_{j}} \L&= [(B_{j})_+,\L],&
  \epsilon\partial_{y_{j}} \L&= [(F_{j})_+,\L],\
  \epsilon\partial_{ t_{j}} \log \L= [(B_{j})_+ ,\log \L],&
  \end{align}
   \begin{align}\epsilon(\log \L)_{ y_{j}}=[ -(F_{j})_-,\log_+ \L ]+
[(F_{j})_+ ,\log_- \L ].
\end{align}
\end{proposition}

\subsection{Strongly coupled extended Toda equations}
 As a consequence of the factorization problem \eqref{facW} and  Sato equations, after taking into account that   $S\in G_{-}$ and $\bar S\in G_{+}$, the $t_0$ flow of $\L$ in the form of $\L=\Lambda+U+V\Lambda^{-1}$ is as
\begin{gather}\label{exp-omega}
\begin{aligned}
  \epsilon\partial_{t_{0}} \L&= [\Lambda+U,V\Lambda^{-1}],
  \end{aligned}
\end{gather}
which lead to a strongly coupled Toda equation
\[\epsilon\partial_{t_{0}} U&=& V(x+\epsilon)-V(x),\\ \label{toda}
\epsilon\partial_{t_{0}} V&=& U(x)V(x)-V(x)U(x-\epsilon).\]
Of course, one can switch the order of the matrices because of the commutativity of $S_2$.
Suppose
\[U=\begin{bmatrix}u_0&u_1\\ u_1&u_0\end{bmatrix},\ \ V=\begin{bmatrix}v_0&v_1\\ v_1&v_0\end{bmatrix},\]
then the specific strongly coupled  Toda equation is
\[\epsilon\partial_{t_{0}} u_0&=& v_0(x+\epsilon)-v_0(x), \\
\epsilon\partial_{t_{0}} u_1&=& v_1(x+\epsilon)-v_1(x),\\ \label{stoda}
\epsilon\partial_{t_{0}} v_0&=& u_0(x)v_0(x)+u_1(x)v_1(x)-v_0(x)v_0(x-\epsilon)-v_1(x)v_1(x-\epsilon),\\
\epsilon\partial_{t_{0}} v_1&=&(u_1(x)-u_1(x-\epsilon))v_0(x)-v_1(x)(u_0(x)-u_0(x-\epsilon).\]
 To get the standard strongly coupled Toda equation, one need to use the alternative expressions
\begin{gather}\label{exp-omega1}
\begin{aligned}
  U_0&:=\omega_1(x)-\omega_1(x+\epsilon)=\epsilon\partial_{t_1}\phi_0(x),\\
  U_1&:=\bar\omega_1(x)-\bar\omega_1(x+\epsilon)=\epsilon\partial_{t_1}\phi_1(x),\\
 V_0&= \Exp{\phi_0(x)-\phi_0(x-\epsilon)}\cosh(\phi_1(x)-\phi_1(x-\epsilon))=-\epsilon\partial_{t_1}\omega_1(x),\\
 V_1&= \Exp{\phi_0(x)-\phi_0(x-\epsilon)}\sinh(\phi_1(x)-\phi_1(x-\epsilon))=-\epsilon\partial_{t_1}\bar\omega_1(x).
\end{aligned}
\end{gather}

From Sato equations we deduce the following set of nonlinear
partial differential-difference equations
\begin{align}\left\{
\begin{aligned}
 \omega_1(x)-\omega_1(x+\epsilon)&=\epsilon\partial_{t_1}\phi_0(x),\\
  \bar\omega_1(x)-\bar\omega_1(x+\epsilon)&=\epsilon\partial_{t_1}\phi_1(x),\\
\epsilon\partial_{t_1}\omega_1(x)&=-\Exp{\phi(x)-\phi(x-\epsilon)}\cosh(\phi_1(x)-\phi_1(x-\epsilon))\\
\epsilon\partial_{t_1}\bar\omega_1(x)&=-\Exp{\phi(x)-\phi(x-\epsilon)}\sinh(\phi_1(x)-\phi_1(x-\epsilon)).\end{aligned}\right.
\label{eq:multitoda}
\end{align}
Observe that if we cross the above equations, then we get
the following strongly coupled Toda system
\begin{align*}
  \epsilon^2\partial_{t_1}^2\phi_0(x)&=
  \Exp{\phi_0(x+\epsilon)-\phi_0(x)}\cosh(\phi_0(x+\epsilon)-\phi_0(x))-\Exp{\phi_0(x)-\phi_0(x-\epsilon)}\cosh(\phi_0(x)-\phi_0(x-\epsilon)),\\
  \epsilon^2\partial_{t_1}^2\phi_1(x)&=
  \sinh(\phi_1(x+\epsilon)-\phi_1(x))\Exp{\phi_0(x+\epsilon)-\phi_0(x)}-  \sinh(\phi_1(x)-\phi_1(x-\epsilon))\Exp{\phi_0(x)-\phi_0(x-\epsilon)}.\\
\end{align*}

Besides above strongly coupled Toda equations, with logarithm flows the SCETH also contains some extended flow equations in the next part.
Here we consider  the extended flow equations in the simplest case, i.e. the $y_{0}$ flow for $\L,$
\[\epsilon \d_{y_{0}}\L&=&[(S\epsilon \d_x S^{-1})_+,\L]\\
&=&[\epsilon \d_xS S^{-1},\L]\\
&=&\epsilon\L_x,\]
which leads to  the following specific equation
\[\d_{y_{0}}U&=& U_{x} ,\ \ \d_{y_{0}}V= V_{x} .\]

To see the extended equations clearly, one need to rewrite the extended flows  in the Lax equations of the SCETH as in the following lemma.

\begin{lemma}\label{modifiedLax}

The extended flows in Lax formulation of the SCETH can be equivalently given
by
\begin{equation}
  \label{edef2}
\epsilon\frac{\partial \L}{\partial y_{j}} = [D_{j} ,\L ],
\end{equation}
\begin{align}
  &D_{j} = (\frac{\L^j}{j!}(\log_+ \L-c_j))_+-(\frac{\L^j}{j!}(\log_- \L-c_j))_-,
\end{align}
which can also be rewritten in the form
\begin{equation}
  \label{edef2'}
\epsilon \frac{\partial \L}{\partial y_n} = [\bar  D_n ,\L ],
\end{equation}
\begin{align} \notag
 \bar D_{j} &= \frac{\L^j}{j!}\epsilon \d+ [ \frac{\L^j}{j!} (\sum_{k < 0} W_k(x) \Lambda^k -c_j)]_+-[ \frac{\L^j}{j!}  (\sum_{k \geq 0} W_k(x) \Lambda^k -c_j)]_-.
\end{align}

\end{lemma}

Then one can derive the $y_1$ flow equation of the SCETH as
\[\notag &\epsilon u_{0 y_1}=\frac\epsilon 2(1-\Lambda)(v_0(\Lambda^{-1}-1)^{-1}(\log (v_0+v_1)+\log (v_0-v_1))_x)-2(\Lambda-1)v_0\\
&+\frac\epsilon 2(1-\Lambda)(v_1(\Lambda^{-1}-1)^{-1}(\log (v_0+v_1)-\log (v_0-v_1))_x)+\frac{\epsilon}2u_{0x}^2+\frac{\epsilon}2u_{1x}^2+\epsilon v_{0x},\\
\notag &\epsilon u_{1 y_1}=\frac\epsilon 2(1-\Lambda)(v_0(\Lambda^{-1}-1)^{-1}(\log (v_0+v_1)-\log (v_0-v_1))_x)-2(\Lambda-1)v_1\\
&+\frac\epsilon 2(1-\Lambda)(v_1(\Lambda^{-1}-1)^{-1}(\log (v_0+v_1)+\log (v_0-v_1))_x)+\epsilon u_{0x}u_{1x}+\epsilon v_{1x},\]
\[
\notag &\epsilon v_{0y_1}=((\Lambda^{-1}-1)^{-1}\frac\epsilon 2 (\log (v_0+v_1)+\log (v_0-v_1))_x)+2)[(u_0(x-\epsilon)-u_0(x))v_0+(u_1(x-\epsilon)-u_1(x))v_1]\\
\notag &(\Lambda^{-1}-1)^{-1}\frac\epsilon 2 (\log (v_0+v_1)-\log (v_0-v_1))_x)[(u_1(x-\epsilon)-u_1(x))v_0+(u_1(x-\epsilon)-u_1(x))v_0]\\
&\notag+\epsilon v_{0x}u_0(x-\epsilon)+\epsilon v_{1x}u_1(x-\epsilon)+\epsilon (u_{0x}(x-\epsilon)+ u_{0x}(x))v_0+\epsilon (u_{1x}(x-\epsilon)+ u_{1x}(x))v_1,\\
\notag &\epsilon v_{1y_1}=((\Lambda^{-1}-1)^{-1}\frac\epsilon 2 (\log (v_0+v_1)-\log (v_0-v_1))_x)+2)[(u_0(x-\epsilon)-u_0(x))v_0+(u_1(x-\epsilon)-u_1(x))v_1]\\
\notag &(\Lambda^{-1}-1)^{-1}\frac\epsilon 2 (\log (v_0+v_1)+\log (v_0-v_1))_x)[(u_1(x-\epsilon)-u_1(x))v_0+(u_1(x-\epsilon)-u_1(x))v_0]\\
&\notag+\epsilon v_{1x}u_0(x-\epsilon)+\epsilon v_{0x}u_1(x-\epsilon)+\epsilon (u_{1x}(x-\epsilon)+ u_{1x}(x))v_0+\epsilon (u_{0x}(x-\epsilon)+ u_{0x}(x))v_1,\]
where $u_i,v_i$ without bracket behind them means $u_i(x),v_i(x)$ respectively.

\section{Virasoro symmetries of the strongly coupled extended Toda hierarchy}

In this section, we will put constrained condition
eq.\eqref{LaxZ} into a construction of the flows of additional
symmetries which form the well-known Virasoro algebra.

With the dressing operators given in eq.\eqref{LaxZ}, we introduce Orlov-Schulman operators as following
\begin{eqnarray}\label{Moperator}
&&\M=P\Gamma P^{-1}, \ \ \bar \M=\bar P\bar \Gamma \bar P^{-1},\ \\
 &&\Gamma=
\frac{x}{\epsilon}\Lambda^{-1}+\sum_{n\geq 0}
n\Lambda^{n-1}t_{n}+\sum_{n\geq 0}
\frac{1}{ (n-1)!}  \Lambda^{n-1} (\epsilon \d -  c_{n-1} )y_{n},\\
&&\bar \Gamma=
\frac{-x}{\epsilon}\Lambda-\sum_{n\geq 0}
\frac{1}{ (n-1)!}  \Lambda^{-n+1} (-\epsilon \d -  c_{n} )y_{n}.
\end{eqnarray}
Therefore we can get
\[\notag \M-\bar \M&=&
P\frac{x}{\epsilon}\Lambda^{-1}P^{-1}+\sum_{n\geq 0}
nt_{n}B_{n-1}+
\bar P\frac{x}{\epsilon}\Lambda \bar P^{-1}+\sum_{n\geq 0}
\frac{2}{ (n-1)!}  \Lambda^{n-1} (\log \L -  c_{n} )y_{n},\]
is a pure difference operator.

Then one can prove the Lax operator $\L$ and Orlov-Schulman operators $\M,\bar \M$ satisfy the following theorem.
\begin{proposition}\label{flowsofMZ}
The $S_2$-valued Lax operator $\L$ and Orlov-Schulman operators $\M,\bar \M$ of the SCETH
satisfy the following
\begin{eqnarray}
&[\L,\M]=\L,[\L,\bar \M]=1,[\log_+\L,M]=S\Lambda^{-1} S^{-1},[\log_-\L,\bar \M]=\bar P \Lambda\bar P^{-1},\\ \label{Mequation}
&\epsilon\partial_{t_{k}} \M^m\L^k= [(B_{k})_+,\M^m\L^k], \epsilon\partial_{\bar t_{k}} \M^m\L^k= [ -(\bar B_{k})_-,\M^m\L^k],\\
&\epsilon\partial_{t_{k}} \bar \M^m\L^k= [(B_{k})_+,\bar \M^m\L^k], \epsilon\partial_{\bar t_{k}} \bar \M^m\L^k= [ -(\bar B_{k})_-,\bar \M^m\L^k],\\
 \notag
&\dfrac{\partial
\M^m\L^k}{\partial{y_{j}}}=[\frac{\L^j}{j!}(\log_+ \L-c_j) -(F_{j})_-,
\M^m\L^k],\;  \dfrac{\partial
\bar \M^m\L^k}{\partial{y_{j}}}=[-\frac{\L^j}{j!}(\log_- \L-c_j)+(F_{j})_+, \bar \M^m\L^k].\\
\end{eqnarray}
\end{proposition}

\begin{proof}
One can prove the proposition by dressing the following several commutative Lie brackets
\begin{eqnarray*}&&[\partial_{ t_{n}}-\Lambda^{n},\Gamma]\\
&=&[\partial_{ t_{n}}-\Lambda^{n},\frac{x}{\epsilon}\Lambda^{-1}+\sum_{n\geq 0}
n\Lambda^{n-1}t_{n}+\sum_{n\geq 0}
\frac{1}{ (n-1)!}  \Lambda^{n-1} (\epsilon \d -  c_{n-1} )y_{n}]\\&=&0,
\end{eqnarray*}

\begin{eqnarray*}&&[\partial_{ y_{n}}-\frac{1}{ n!}  \Lambda^n (\epsilon \d -  c_n ),\Gamma]\\
&=&[\partial_{ y_{n}}-\frac{1}{ n!}  \Lambda^n (\epsilon \d -  c_n ),\frac{x}{\epsilon}\Lambda^{-1}+\sum_{n\geq 0}
n\Lambda^{n-1}t_{n}+\sum_{n\geq 0}
\frac{1}{ (n-1)!}  \Lambda^{n-1} (\epsilon \d -  c_{n-1} )y_{n}]\\&=&0,
\end{eqnarray*}

\begin{eqnarray*}&&[\partial_{ y_{n}}+\frac{1}{ n!}  \Lambda^{-n} (-\epsilon \d -  c_n ),\bar \Gamma]\\
&=&[\partial_{ y_{n}}+\frac{1}{ n!}  \Lambda^{-n} (-\epsilon \d -  c_n ),\frac{-x}{\epsilon}\Lambda-\sum_{n\geq 0}
n\Lambda^{-n+1}\bar t_{n}-\sum_{n\geq 0}
\frac{1}{ (n-1)!}  \Lambda^{-n+1} (-\epsilon \d -  c_{n} )y_{n}]\\&=&0.
\end{eqnarray*}

\end{proof}
We are now to define the additional flows, and then to
prove that they are symmetries, which are called additional
symmetries of the SCETH. We introduce additional
independent variables $t^*_l$ and define the actions of the
additional flows on the wave operators as
\begin{eqnarray}\label{definitionadditionalflowsonphi2Z}
\dfrac{\partial P}{\partial
{t^*_l}}=-\left((\M-\bar \M)\L^l\right)_{-}P, \ \ \ \dfrac{\partial
\bar P}{\partial {t^*_l}}=\left((\M-\bar \M)\L^l\right)_{+}\bar P,
\end{eqnarray}
where $m\geq 0, l\geq 0$.

Then we can derive the following proposition.
\begin{proposition}\label{add flowZ}
The additional derivatives  act on  $\M$, $\bar \M$ as
\begin{eqnarray}
\label{SCETHadditionalflow11'}
\dfrac{\partial
\M}{\partial{t^*_l}}&=&[-\left((\M-\bar \M)\L^l\right)_{-}, \M],
\\
\label{SCETHadditionalflow12}
\dfrac{\partial
\bar \M}{\partial{t^*_l}}&=&[\left((\M-\bar \M)\L^l\right)_{+}, \bar \M].
\end{eqnarray}
\end{proposition}

By the propositions above, we can find for $ n,k,l\geq 0$,
the following identities hold
\begin{eqnarray}\label{SCETHadditionalflow4}
 \dfrac{\partial \M^n\L^k}{\partial{t^*_l}}=-[((\M-\bar \M)\L^l)_{-}, \M^n\L^k]
,\ \ \
 \dfrac{\partial \bar \M^n\L^k}{\partial{t^*_l}}=[((\M-\bar \M)\L^l)_{+},
 \bar \M^n\L^k].
\end{eqnarray}

Basing on above results, the following theorem can be proved.

\begin{theorem}\label{symmetryz}
The additional flows $\partial_{t^*_l}$ commute
with the SCETH flows, i.e.,
\begin{eqnarray}
[\partial_{t^*_l}, \partial_{t_{n}}]\Phi=0,\ \  \ [\partial_{t^*_l}, \partial_{y_{n}}]\Phi=0,
\end{eqnarray}
where $\Phi$ can be $P$, $\bar P$ or $\L$, $1\leq\gamma\leq N; n\geq 0$ and
 $
\partial_{t^*_l}=\frac{\partial}{\partial{t^*_l}},
\partial_{t_{n}}=\frac{\partial}{\partial{t_{n}}}$.

\end{theorem}
\begin{proof}  Here
we also give the proof for commutativity of  additional symmetries with the extended
flow $\d_{y_{n}}$. To be an example, we only let the Lie bracket act on $\bar P$,
\begin{eqnarray*}
[\partial_{t^*_l},\partial_{y_{n}}]\bar P &=&
\partial_{t^*_l}(F_{j})_+\bar P -
\partial_{y_{n}} \left(((\M-\bar \M)\L^l)_{+}\bar P  \right)\\
&=& \partial_{t^*_l}(F_{j})_+ \bar P +
(F_{j})_+ (\partial_{t^*_l}
)\bar P\\&&-(\partial_{y_{n}} ((\M-\bar \M)\L^l))_{+}\bar P
-((\M-\bar \M)\L^l)_{+}(\partial_{y_{n}}\bar P ),
\end{eqnarray*}
which further leads to
\begin{eqnarray*}
[\partial_{t^*_l},\partial_{y_{n}}]\bar P
&=&[((\M-\bar \M)\L^l)_{+}, \frac{\L^j}{j!}(\log_- \L-c_j)]_{+}\bar P\\
&&-[((\M-\bar \M)\L^l)_{-}, \frac{\L^j}{j!}(\log_+ \L-c_j)]_{+}\bar P\\
&&+(F_{j})_+ ((\M-\bar \M)\L^l)_{+}\bar P -((\M-\bar \M)\L^l)_{+}(F_{j})_+\bar P\\
&&-[\left(\frac{\L^j}{j!}(\log_+ \L-c_j) \right)_+-\left(\frac{\L^j}{j!}(\log_- \L-c_j) \right)_-,(\M-\bar \M)\L^l]_{+}\bar P  \\
&=&[((\M-\bar \M)\L^l)_{+}, (\frac{\L^j}{j!}(\log_- \L-c_j))_+]_{+}\bar P\\
&&-[((\M-\bar \M)\L^l)_{-}, \frac{\L^j}{j!}(\log_+ \L-c_j)]_{+}\bar P\\
&&+(F_{j})_+ ((\M-\bar \M)\L^l)_{+}\bar P -((\M-\bar \M)\L^l)_{+}(F_{j})_+\bar P\\
&&+[(\M-\bar \M)\L^l,\left(\frac{\L^j}{j!}(\log_+ \L-c_j) \right)_+]_{+}\bar P  \\
&=&0.
\end{eqnarray*} The other cases in the theorem can be proved in similar ways.
\end{proof}
The commutative property in Theorem \ref{symmetryz} means that
additional flows are symmetries of the SCETH.
As a special reduction from the EMTH to the SCETH, it is easy to derive the algebraic
structures among these additional symmetries in the following important
theorem.
\begin{theorem}\label{WinfiniteCalgebra}
The additional flows  $\partial_{t^*_l}$ of the SCETH form a Virasoro type Lie algebra with the
following relation
 \begin{eqnarray}\label{algebra relation}
[\partial_{t^*_l},\partial_{t^*_k}]= (k- l)\d^*_{k+l-1},
\end{eqnarray}
which holds in the sense of acting on  $P$, $\bar P$ or $\L$ and  $l,k\geq 0.$
\end{theorem}

\section{Virasoro action on tau-functions of SCETH}

Introduce the following sequence:
\[t-[\lambda] &:=& (t_{j}-
  \epsilon(j-1)!\lambda^j, 0\leq j\leq \infty).
\]
A $S_2$-valued function $\begin{bmatrix}\tau& \sigma\\ \sigma & \tau\end{bmatrix}\in S_2$  depending only on the dynamical variables $t$ and
$\epsilon$ is called the  $S_2$-valued tau-function of the SCETH if it
provides symbols related to matrix-valued wave operators as following,

\begin{eqnarray}\label{Mpltaukk}\bP_0: &=&\frac{ \tau
(t-[\lambda^{-1}])\tau(t)- \sigma
(t-[\lambda^{-1}])\sigma(t)}
     {\tau^2(t)-\sigma^2(t)},\\
     \bP_1: &=&\frac{\sigma
(t-[\lambda^{-1}])\tau(t)- \tau
(t-[\lambda^{-1}])\sigma(t) }
     {\tau^2(t)-\sigma^2(t)},\\
     \label{Mpl-1taukk}\bP_2: &=&\frac{ \tau
(x+\epsilon,t+[\lambda^{-1}])\tau(x+\epsilon,t)- \sigma
(x+\epsilon,t+[\lambda^{-1}])\sigma(x+\epsilon,t)}
     {\tau^2(x+\epsilon,t)-\sigma^2(x+\epsilon,t)},\\
         \label{Mpl-1taukk}\bP_3: &=&\frac{ \sigma
(x+\epsilon,t+[\lambda^{-1}])\sigma(x+\epsilon,t)-\tau
(x+\epsilon,t+[\lambda^{-1}])\tau(x+\epsilon,t)}
     {\tau^2(x+\epsilon,t)-\sigma^2(x+\epsilon,t)},\\
     \label{Mpltaukk}\bar\bP_0: &=&\frac{ \tau
(x+\epsilon,t+[\lambda])\tau(t)- \sigma
(x+\epsilon,t+[\lambda])\sigma(t)}
     {\tau^2(t)-\sigma^2(t)},\\
     \bar\bP_1: &=&\frac{\sigma
(x+\epsilon,t+[\lambda])\tau(t)- \tau
(x+\epsilon,t+[\lambda])\sigma(t) }
     {\tau^2(t)-\sigma^2(t)},\\
     \label{Mpl-1taukk}\bar\bP_2: &=&\frac{ \tau
(x,t-[\lambda^{-1}])\tau(x+\epsilon,t)- \sigma
(x,t-[\lambda^{-1}])\sigma(x+\epsilon,t)}
     {\tau^2(x+\epsilon,t)-\sigma^2(x+\epsilon,t)},\\
         \label{Mpl-1taukk}\bar\bP_3: &=&\frac{ \sigma
(x,t-[\lambda^{-1}])\sigma(x+\epsilon,t)-\tau
(x,t-[\lambda^{-1}])\tau(x+\epsilon,t)}
     {\tau^2(x+\epsilon,t)-\sigma^2(x+\epsilon,t)}.
     \end{eqnarray}

Then  according to the ASvM formula in \cite{asv2} and a commutative algebraic reduction, we can get the following formula
\[\frac{(\d_{t^*_{k}}\bP_0) \bP_0-(\d_{t^*_{k}}\bP_1) \bP_1}{\bP_0^2-\bP_1^2}&=&(e^{-\sum_{i=1}^\infty \epsilon(i-1)!\lambda^{-i}\d_{t_i}}-1)\frac{(L_{k-1}\tau)\tau-(L_{k-1}\sigma)\sigma}{\tau^2-\sigma^2},\\
\frac{(\d_{t^*_{k}}\bP_1) \bP_0-(\d_{t^*_{k}}\bP_0) \bP_1}{\bP_0^2-\bP_1^2}&=&(e^{-\sum_{i=1}^\infty \epsilon(i-1)!\lambda^{-i}\d_{t_i}}-1)\frac{(L_{k-1}\sigma)\tau-(L_{k-1}\tau)\sigma}{\tau^2-\sigma^2},\\
\frac{(\d_{t^*_{k}}\bar\bP_0) \bar\bP_0-(\d_{t^*_{k}}\bar\bP_1) \bar\bP_1}{\bar\bP_0^2-\bar\bP_1^2}&=&(e^{\epsilon\d_x+\sum_{i=1}^\infty \epsilon(i-1)!\lambda^i\d_{t_i}}-1)\frac{(L_{k-1}\tau)\tau-(L_{k-1}\sigma)\sigma}{\tau^2-\sigma^2},\\
\frac{(\d_{t^*_{k}}\bar\bP_1) \bar\bP_0-(\d_{t^*_{k}}\bar\bP_0) \bar\bP_1}{\bar\bP_0^2-\bar\bP_1^2}&=&(e^{\epsilon\d_x+\sum_{i=1}^\infty \epsilon(i-1)!\lambda^i\d_{t_i}}-1)\frac{(L_{k-1}\sigma)\tau-(L_{k-1}\tau)\sigma}{\tau^2-\sigma^2},
\]
where
\[L_{-1}&=&\sum_{n=1}^\infty(t_{n-1}\d_{t_{n-2}}+2\frac{c_{n-1}}{(n-1)!}y_n\d_{t_{n-1}}-2c_{n-1}y_n\d_{t_{n-2}}+y_n\d_{y_{n-1}})+t_0y_0,\\
L_{0}&=&\sum_{n=1}^\infty(nt_{n-1}\d_{t_{n-1}}+2\frac{c_n}{(n-1)!}y_n\d_{t_{n}}-2c_{n-1}y_n\d_{t_{n-1}}+ny_n\d_{y_{n}})+y_0^2,\\ \notag
L_{p}&=&\sum_{n=1}^\infty(\frac{(n+p)!}{(n-1)!}t_{n-1}\d_{t_{n+p-1}}+2\frac{c_{n+p}}{(n-1)!}y_n\d_{t_{n+p}}-2\frac{(n+p)!}{(n-1)!}c_{n-1}y_n\d_{t_{n+p-1}}\\
&&+\frac{(n+p)!}{(n-1)!}y_n\d_{y_{n+p}})+2p!y_0\d_{t_{p-1}}+\sum_{n=1}^{p-1}n!(p-n)!\d_{t_{n-1}}\d_{t_{p-n-1}},\ \  p\geq 1.\]

These operators $\{L_k,\ k\geq -1\}$ constitute a Virasoro algebra \cite{cmp,thesis} (one half without the cental extension) as
\[[L_m,L_n]=(m-n)L_{m+n}.\]
The central extension appears only if we consider the action on the
tau function as it was done in \cite{CT1,CT2}.

\section{Multi-fold Darboux transformations of the SCETH}
In this section, we will consider the Darboux transformation of the SCETH on the Lax operator
  \[\label{1darbouxL}L^{[1]}=WL W^{-1},\]
where $W$ is the Darboux transformation operator.

That means after the Darboux transformation, the spectral problem

\[L\phi=\Lambda\phi+u\phi+ v
\Lambda^{-1}\phi=\lambda\phi,\]
will become

\[L^{[1]}\phi^{[1]}=\Lambda\phi^{[1]}+u^{[1]}\phi^{[1]}+v^{[1]}
\Lambda^{-1}\phi^{[1]}=\lambda\phi^{[1]}.\]

To keep the Lax equation of the SCETH invariant, i.e.
 The  Lax equations of the SCETH are as follows
   \begin{align}
\label{laxtjkS}
  \epsilon\partial_{t_{j}} \L^{[1]}&= [(B_{j}^{[1]})_+,\L^{[1]}],&
  \epsilon\partial_{y_{j}} \L^{[1]}&= [(F_{j}^{[1]})_+,\L^{[1]}],
  \end{align}
\begin{equation}
 \ B_{j}^{[1]}:=B_{j}(\L^{[1]}),\ \ F_{j}^{[1]}:=F_{j}(\L^{[1]}),
\end{equation} the dressing operator $W$ should satisfy the following  equation
\[W_{t_{j}}=-W(B_{j})_++(WB_{j}W^{-1})_+W,\ W_{y_{j}}=-W(F_{j})_++(WF_{j}W^{-1})_+W, \ j\geq 0.\]

Now, we will give the following important theorem which will be used to generate new solutions.

\begin{theorem}
If $\phi=\begin{bmatrix}\phi_0&\phi_1\\ \phi_1&\phi_0\end{bmatrix}$ is the first wave function of the SCETH,
the Darboux transformation operator of the SCETH
 \[W(\lambda)=(\I_2-\frac{\phi}{\La^{-1}\phi}\La^{-1})=\phi\circ(\I_2-\La^{-1})\circ\phi^{-1},\]

will generater new solutions

\[\label{1uN-11}u_0^{[1]}&=&u_0+(\La-\I_2)\frac{\phi_0(x)\phi_0(x-\epsilon)-\phi_1(x)\phi_1(x-\epsilon)}{\phi_0^2(x-\epsilon)-\phi_1^2(x-\epsilon)},\\
u_1^{[1]}&=&u_1+(\La-\I_2)\frac{\phi_1(x)\phi_0(x-\epsilon)-\phi_0(x)\phi_1(x-\epsilon)}{\phi_0^2(x-\epsilon)-\phi_1^2(x-\epsilon)},\]

{\tiny\[\notag
&&v_0^{[1]}=\frac{(\phi_0(x)\phi_0(x-2\epsilon)+\phi_1(x)\phi_1(x-2\epsilon))(\phi_0^2(x-\epsilon)+\phi_1^2(x-\epsilon))-(\phi_0(x)\phi_1(x-2\epsilon)
+\phi_1(x)\phi_0(x-2\epsilon))(2\phi_0(x-\epsilon)\phi_1(x-\epsilon))}
{(\phi_0^2(x-\epsilon)+\phi_1^2(x-\epsilon))^2-(2\phi_0(x-\epsilon)\phi_1(x-\epsilon))^2}v_0(x-\epsilon)\\ \notag
&&+\frac{(\phi_0(x)\phi_1(x-2\epsilon)
+\phi_1(x)\phi_0(x-2\epsilon))(\phi_0^2(x-\epsilon)+\phi_1^2(x-\epsilon))-(\phi_0(x)\phi_0(x-2\epsilon)+\phi_1(x)\phi_1(x-2\epsilon))(2\phi_0(x-\epsilon)\phi_1(x-\epsilon))}
{(\phi_0^2(x-\epsilon)+\phi_1^2(x-\epsilon))^2-(2\phi_0(x-\epsilon)\phi_1(x-\epsilon))^2}v_1(x-\epsilon),\\ \notag
&&v_1^{[1]}=\frac{(\phi_0(x)\phi_0(x-2\epsilon)+\phi_1(x)\phi_1(x-2\epsilon))(\phi_0^2(x-\epsilon)+\phi_1^2(x-\epsilon))-(\phi_0(x)\phi_1(x-2\epsilon)
+\phi_1(x)\phi_0(x-2\epsilon))(2\phi_0(x-\epsilon)\phi_1(x-\epsilon))}
{(\phi_0^2(x-\epsilon)+\phi_1^2(x-\epsilon))^2-(2\phi_0(x-\epsilon)\phi_1(x-\epsilon))^2}v_1(x-\epsilon)\\ \notag
&&+\frac{(\phi_0(x)\phi_1(x-2\epsilon)
+\phi_1(x)\phi_0(x-2\epsilon))(\phi_0^2(x-\epsilon)+\phi_1^2(x-\epsilon))-(\phi_0(x)\phi_0(x-2\epsilon)+\phi_1(x)\phi_1(x-2\epsilon))(2\phi_0(x-\epsilon)\phi_1(x-\epsilon))}
{(\phi_0^2(x-\epsilon)+\phi_1^2(x-\epsilon))^2-(2\phi_0(x-\epsilon)\phi_1(x-\epsilon))^2}v_0(x-\epsilon).\]}

\end{theorem}

Define $ \phi_i=\phi_i^{[0]}:=\phi|_{\lambda=\lambda_i}=\begin{bmatrix}\phi_{i0}&\phi_{i1}\\ \phi_{i1}&\phi_{i0}\end{bmatrix}$, then after iteration on Darboux transformations, we can generalize the Darboux transformation to the $n$-fold case.

Taking seed solution $u_0=u_1=0,v_0=1,v_1=0$, then after iteration on Darboux transformations,  one can get the $n$-th new solution of the SCETH as

\[u^{[n]}&=&\frac12(1-\La^{-1})\d_{t_{2,0}}(\log (\tau_n+\sigma_n)+\log (\tau_n-\sigma_n)),\\
v^{[n]}&=&e^{\frac12(1-\La^{-1})^2(\log (\tau_n+\sigma_n)-\log (\tau_n-\sigma_n))},\]
\[\tau_n&=&\sum_{m=0}^{[\frac n2]}\sum_{ \sum_{i=1}^n\alpha_i=2m}Wr (\phi_1^{(\alpha_1)},\phi_2^{(\alpha_2)},\dots\phi_n^{(\alpha_n)}),\  \\
\sigma_n&=& \sum_{m=0}^{[\frac n2]}\sum_{ \sum_{i=1}^n\alpha_i=2m+1}Wr (\phi_1^{(\alpha_1)},\phi_2^{(\alpha_2)},\dots\phi_n^{(\alpha_n)}),\  \alpha_i=0,1; \phi_i^{(0)}=\phi_{i0},\ \phi_i^{(1)}=\phi_{i1},
\]

where $Wr (\phi_1^{(\alpha_1)},\phi_2^{(\alpha_2)},\dots\phi_n^{(\alpha_n)})$ is the Wronskian, i.e. a Casorati determinant
\[Wr (\phi_1^{(\alpha_1)},\phi_2^{(\alpha_2)},\dots\phi_n^{(\alpha_n)})=det (\La^{-j+1} \phi_{n+1-i}^{(\alpha_{n+1-i})})_{1\leq i,j\leq n},\]
Particularly for the SCETH, choosing appropriate wave function $\phi$, the $n$-th new solutions can be solitary wave solutions, i.e. $n$-soliton solutions.

{\bf {Acknowledgements:}}
  This work is funded by the National Natural Science Foundation of China under Grant No.  11571192,  and K. C. Wong Magna Fund in
Ningbo University.


\end{document}